\documentclass[journal]{IEEEtran} 		% For final IEEE paper version
\usepackage{amsmath,amssymb}
\usepackage{amsthm}
\usepackage{color}
\usepackage[mathscr]{eucal}
\usepackage{graphics,graphicx,multicol}
\usepackage{epstopdf}
\usepackage{enumerate}
\usepackage{subfigure}
\usepackage{algorithmic}
\usepackage[section]{algorithm}
\usepackage{morefloats}
\usepackage{multirow}
\usepackage{array}
\usepackage{pstricks, pst-node, pst-plot, pst-circ}
\usepackage{moredefs}
\usepackage{courier}
\usepackage{colortbl}% http://ctan.org/pkg/colortbl
\usepackage{xcolor}% http://ctan.org/pkg/xcolor
\usepackage{cite}

\newtheorem{thm}{Theorem}[section]
\newtheorem{lem}[thm]{Lemma}
\newtheorem{cor}[thm]{Corollary}
\newtheorem{rem}[thm]{Remark}
\begin{document}
\title{Robust Resource Allocation with Joint Carrier Aggregation for Multi-Carrier Cellular Networks}
\author{Haya Shajaiah, Ahmed Abdelhadi, and T. Charles Clancy \\
Bradley Department of Electrical and Computer Engineering\\
Hume Center, Virginia Tech, Arlington, VA, 22203, USA\\
\{hayajs, aabdelhadi, tcc\}@vt.edu
\thanks{H. Shajaiah, A. Abdelhadi, and C. Clancy are with the Hume Center for National Security and Technology, Virginia Tech, Arlington,
VA, 22203 USA e-mail: \{hayajs, aabdelhadi, tcc\}@vt.edu. This paper is an extension of IEEE ICNC Workshop CCS 2015 paper titled, "An Optimal Resource Allocation with Joint Carrier Aggregation in 4G-LTE". This paper considers the case of scarce resources with respect to the number of users and their traffic. The robust distributed resource allocation algorithm presented in this paper solves the drawback in the algorithm presented in the conference paper, by preventing the fluctuations in the resource allocation process in the case of scarce resources.}
}
%\thanks{Manuscript received April 19, 2014; revised May 27, 2014.}

%\markboth{Journal of \LaTeX\ Class Files,~Vol.~11, No.~4, December~2014}%
%{Shell \MakeLowercase{\textit{et al.}}: Bare Demo of IEEEtran.cls for Journals}

\maketitle
\begin{abstract}
In this paper, we present a novel approach for robust optimal resource allocation with joint carrier aggregation to allocate multiple carriers resources optimally among users with elastic and inelastic traffic in cellular networks. We use utility proportional fairness allocation policy, where the fairness among users is in utility percentage of the application running on the user equipment (UE). Each UE is assigned an application utility function based on the type of its application. Our objective is to allocate multiple carriers resources optimally among users subscribing for mobile services. In addition, each user is guaranteed a minimum quality of service (QoS) that varies based on the user's application type.
We present a robust algorithm that solves the drawback in the algorithm presented in \cite{Ahmed_Utility4} by preventing the fluctuations in the resource allocation process, in the case of scarce resources, and allocates optimal rates for both high-traffic and low-traffic situations.
Our distributed resource allocation algorithm allocates an optimal rate to each user from all carriers in its range while providing the minimum price for the allocated rate. In addition, we analyze the convergence of the algorithm with different network traffic densities and show that our algorithm provides traffic dependent pricing for network providers. Finally, we present simulation results for the performance of our resource allocation algorithm.
\end{abstract}
% Global Parameters that can be changed:
%%%%%%%%%%%%%%%%%%%%%%%%%%%%%%%%%%%
\providelength{\AxesLineWidth}       \setlength{\AxesLineWidth}{0.5pt}%
\providelength{\plotwidth}           \setlength{\plotwidth}{8cm}% width of the axes only
\providelength{\LineWidth}           \setlength{\LineWidth}{0.7pt}%
\providelength{\MarkerSize}          \setlength{\MarkerSize}{3pt}%
\newrgbcolor{GridColor}{0.8 0.8 0.8}%
\newrgbcolor{GridColor2}{0.5 0.5 0.5}%
%%%%%%%%%%%%%%%%%%%%%%%%%%%%%%%%%%%
%%%%%%%%%%%%%%%%%%%%%%%%%%%%%%%%%%%
\section{Introduction}\label{sec:intro}

%$\footnote{This paper is an extension of IEEE ICNC Workshop CCS 2015 paper titled, "An Optimal Resource Allocation with Joint Carrier Aggregation in 4G-LTE". The conference paper has .... So, this journal paper has been expanded ....}
%
In recent years, the number of mobile subscribers and their traffic have increased rapidly. Mobile subscribers are currently running multiple applications, simultaneously, on their smart phones that require a higher bandwidth and make users so limited to the carrier resources. Network providers are now offering multiple services such as multimedia telephony and mobile-TV \cite{QoS_3GPP}. More spectrum is required to meet these demands \cite{Carrier_Agg_1}. However, it is difficult to provide the required resources with a single frequency band due to the scarcity of the available radio spectrum. Therefore, aggregating different carriers' frequency bands is needed to utilize the radio resources across multiple carriers and allow a scalable expansion of the effective bandwidth delivered to the user terminal, leading to interband non-contiguous carrier aggregation \cite{Carrier_Agg_2}.

Carrier aggregation (CA) is one of the most distinct features of 4G systems including Long Term Evolution Advanced (LTE Advanced). Given the fact that LTE requires wide carrier bandwidths to utilize such as $10$ and $20$ MHz, CA needs to be taken into consideration when designing the system to overcome the spectrum scarcity challenges. With the CA being defined in \cite{work-item}, two or more component carriers (CCs) of the same or different bandwidths can be aggregated to achieve wider transmission bandwidths between the evolve node B (eNodeB) and the UE. An overview of CA framework and cases is presented in \cite{CA-framework}. Many operators are willing to add the CA feature to their plans across a mixture of macro cells and small cells. This will provide capacity and performance benefits in areas where small cell coverage is available while enabling network operators to provide robust mobility management on their macro cell networks.

Increasing the utilization of the existing spectrum can significantly improve network capacity, data rates and user experience. Some spectrum holders such as government users do not use their entire allocated spectrum in every part of their geographic boundaries most of the time. Therefore, the National Broadband Plan (NBP) and the findings of the President's Council of Advisors on Science and Technology (PCAST) spectrum study have recommended making the under-utilized federal spectrum available for secondary use \cite{PCAST}. Spectrum sharing enables wireless systems to harvest underutilized swathes of spectrum, which would vastly increase the efficiency of spectrum usage. Making more spectrum available can provide significant gain in mobile broadband capacity only if those resources can be aggregated efficiently with the existing commercial mobile system resources.

This non-contiguous carrier aggregation task is a challenging. The challenges are both in hardware implementation and joint optimal resource allocation. Hardware implementation challenges are in the need for multiple oscillators, multiple RF chains, more powerful signal processing, and longer battery life \cite{RebeccaThesis}. In order to allocate different carriers resources optimally among mobile users in their coverage areas, a distributed resource allocation algorithm between the UEs and the eNodeBs is needed.

A multi-stage resource allocation (RA) with carrier aggregation algorithms are presented in \cite{Haya_Utility1,Haya_Utility3,Haya_Utility6}. The algorithm in \cite{Haya_Utility1} uses utility proportional fairness approach to allocate the primary and the secondary carriers resources optimally among mobile users in their coverage area. The primary carrier first allocates its resources optimally among users in its coverage area. The secondary carrier then starts allocating optimal rates to users in its coverage area based on the users applications and the rates allocated to them by the primary carrier. A RA with CA optimization problem is presented in \cite{Haya_Utility3} to allocate resources from the LTE Advanced carrier and the MIMO radar carrier to each UE, in a LTE Advanced cell based on the application running on the UE. A price selective centralized RA with CA algorithm is presented in \cite{Haya_Utility6} to allocate multiple carriers resources optimally among users while giving the user the ability
to select one of the carriers to be its primary carrier and the others to be its secondary carriers. The UE's decision is based on the carrier price per unit bandwidth. However, the multi-stage RA with CA algorithms presented in \cite{Haya_Utility1,Haya_Utility3,Haya_Utility6} guarantee optimal rate allocation but not optimal pricing.

In this paper, we focus on solving the problem of utility proportional fairness optimal RA with joint CA for multi-carrier cellular networks. The RA with joint CA algorithm presented in \cite{Ahmed_Utility4} fails to converge for high-traffic situations due to the fluctuation in the RA process. In this paper, we present a robust algorithm that solves the drawbacks in \cite{Ahmed_Utility4} and allocates multiple carriers resources optimally among UEs in their coverage area for both high-traffic and low-traffic situations. Additionally, our proposed distributed algorithm outperforms the multi-stage RA with CA algorithms presented in \cite{Haya_Utility1,Haya_Utility3,Haya_Utility6} as it guarantees that mobile users are assigned optimal (minimum) price for resources. We formulate the multi-carrier RA with CA optimization problem into a convex optimization framework. We use logarithmic and sigmoidal-like utility functions to represent delay-tolerant and real-time applications, respectively, running on the mobile
users' smart phones \cite{Ahmed_Utility1}. Our model supports both contiguous and non-contiguous carrier aggregation from one or more network providers. During the resource allocation process, our distributed algorithm allocates optimal resources from one or more carriers to provide the lowest resource price for the mobile users. In addition, we use a utility proportional fairness approach that ensures non-zero resource allocation for all users and gives real-time applications priority over delay-tolerant applications due to the nature of their applications that require minimum encoding rates.

\subsection{Related Work}\label{sec:related}

There has been several works in the area of resource allocation optimization to utilize the scarce radio spectrum efficiently.
The authors in \cite{kelly98ratecontrol,Internet_Congestion,Optimization_flow,Fair_endtoend} have used a strictly concave utility function to represent each user's elastic traffic and proposed distributed algorithms at the sources and the links to interpret the congestion control of communication networks. Their work have only focussed on elastic traffic and did not consider real-time applications as it have non-concave utility functions as shown in \cite{fundamental_design}. The authors in \cite{Utility_max-min} and \cite{ Fair_allocation} have argued that the utility function, which represents the user application performance, is the one that needs to be shared fairly rather than the bandwidth. In this paper, we consider using resource allocation to achieve a utility proportional fairness that maximizes the user satisfaction. If a bandwidth proportional fairness is applied through a max-min bandwidth allocation, users running delay-tolerant applications receive larger utilities than users running real-time
applications as real-time applications require minimum encoding rates and their utilities are equal to zero if they do not receive their minimum encoding rates.

The proportional fairness framework of Kelly introduced in \cite{kelly98ratecontrol} does not guarantee a minimum QoS for each user application. To overcome this issue, a resource allocation algorithm that uses utility proportional fairness policy is introduced in \cite{Ahmed_Utility1}. We believe that this approach is more appropriate as it respects the inelastic behavior of real-time applications. The utility proportional fairness approach in \cite{Ahmed_Utility1} gives real-time applications priority over delay tolerant applications when allocating resources and guarantees that no user is allocated zero rate. In \cite{Ahmed_Utility1, Ahmed_Utility2} and \cite{ Ahmed_Utility3}, the authors have presented optimal resource allocation algorithms to allocate single carrier resources optimally among mobile users. However, their algorithms do not support multi-carrier resource allocation. To incorporate the carrier aggregation feature, we have introduced a multi-stage resource allocation using carrier
aggregation in \cite{Haya_Utility1}. In \cite{Haya_Utility2} and \cite{Haya_Utility4}, we present resource allocation with users discrimination algorithms to allocate the eNodeB resources optimally among mobile users with elastic and inelastic traffic. In \cite{Mo_ResourceBlock}, the authors have presented a radio resource block allocation optimization problem using a utility proportional fairness approach. The authors in \cite{Tugba_ApplicationAware} have presented an application-aware resource block scheduling approach for elastic and inelastic
adaptive real-time traffic where users are assigned to resource blocks.

On the other hand, resource allocation for single cell multi-carrier systems have been given extensive attention in recent years \cite{Dual-Decomposition, Resource_allocation, Rate_Balancing}. In \cite{Fair_resource,Design_of_Fair,Fast_Algorithms,Optimal_and_near-optimal}, the authors have represented this challenge in optimization problems. Their objective is to maximize the overall cell throughput with some constraints such as fairness and transmission power. However, transforming the problem into a utility maximization framework can achieve better users satisfaction rather than better system-centric throughput. Also, in practical systems, the challenge is to perform multi-carrier radio resource allocation for multiple cells. The authors in \cite{Downlink_dynamic,Centralized_vs_Distributed} suggested using a distributed resource allocation rather than a centralized one to reduce the implementation complexity. In \cite{Cooperative_Fair_Scheduling}, the authors propose a collaborative scheme in a multiple
base stations (BSs) environment, where each user is served by the BS that has the best channel gain with that user. The authors in \cite{DownlinkRadio} have addressed the problem of spectrum resource allocation in carrier aggregation based LTE Advanced systems, with the consideration of UEs’ MIMO capability and the modulation and coding schemes (MCSs) selection.

\subsection{Our Contributions}\label{sec:contributions}
Our contributions in this paper are summarized as:
\begin{itemize}
\item We consider the RA optimization problem with joint CA presented in \cite{Ahmed_Utility4} that uses utility proportional fairness approach and solves for logarithmic and sigmoidal-like utility functions representing delay-tolerant and real-time applications, respectively.
\item We prove that the optimization problem is convex and therefore the global optimal solution is tractable. In addition, we present a robust distributed resource allocation algorithm to solve the optimization problem and provide optimal rates in high-traffic and low-traffic situations. % and its simulation results.
\item Our proposed algorithm outperforms that presented in \cite{Ahmed_Utility4} by preventing the fluctuations in the RA process when the resources are scarce with respect to the number of users. It also outperforms the algorithms presented in \cite{Haya_Utility1,Haya_Utility3,Haya_Utility6} as it guarantees that mobile users receive optimal price for resources.
\item We present simulation results for the performance of our RA algorithm and compare it with the performance of the multi-stage RA algorithm presented in \cite{Haya_Utility1,Haya_Utility3,Haya_Utility6}.% In addition, we compare the performance of our algorithm that uses utility proportional fairness policy with the performance of the same algorithm if a bandwidth proportional fairness is used instead.
\end{itemize}

The remainder of this paper is organized as follows. Section \ref{sec:Problem_formulation} presents the problem formulation. Section \ref{sec:Proof} proves that the global optimal solution exists and is tractable. In Section \ref{sec:Dual}, we discuss the conversion of the primal optimization problem into a dual problem. Section \ref{sec:Algorithm} presents our distributed resource allocation algorithm with joint carrier aggregation for the utility proportional fairness optimization problem. In Section \ref{sec:conv_analy}, we present convergence analysis for the allocation algorithm and a modification for robustness at peak-traffic hours. In section \ref{sec:sim}, we discuss simulation setup, provide quantitative results along with discussion and compare the performance of the proposed algorithm with the one presented in \cite{Haya_Utility1,Haya_Utility3,Haya_Utility6}. Section \ref{sec:conclude} concludes the paper.

\section{Problem Formulation}\label{sec:Problem_formulation}

We consider LTE mobile system consisting of $K$ carriers eNodeBs with $K$ cells and $M$ UEs distributed in these cells. The rate allocated by the $l^{th}$ carrier eNodeB to $i^{th}$ UE is given by $r_{li}$ where $l =\{1,2, ..., K\}$ and $i = \{1,2, ...,M\}$. Each UE has its own utility function $U_i(r_{1i}+r_{2i}+ ...+r_{Ki})$ that corresponds to the type of traffic being handled by the $i^{th}$ UE. Our objective is to determine the optimal rates that the $l^{th}$ carrier eNodeB should allocate to the nearby UEs. We express the user satisfaction with its provided service using utility functions that represent the degree of satisfaction of the user function with the rate allocated by the cellular network \cite{DL_PowerAllocation} \cite{fundamental_design} \cite{UtilityFairness}. We assume the utility functions $U_i(r_{1i}+r_{2i}+ ...+r_{Ki})$ to be a strictly concave or a sigmoidal-like functions. The utility functions have the following properties:

\begin{itemize}
\item $U_i(0) = 0$ and $U_i(r_{1i}+r_{2i}+ ...+r_{Ki})$ is an increasing function of $r_{li}$ for $l$.
\item $U_i(r_{1i}+r_{2i}+ ...+r_{Ki})$ is twice continuously differentiable in $r_{li}$ for all $l$.
\end{itemize}
In our model, we use the normalized sigmoidal-like utility function, as in \cite{DL_PowerAllocation}, that can be expressed as
\begin{equation}\label{eqn:sigmoid}
U_i(r_{1i}+r_{2i}+ ...+r_{Ki}) = c_i\Big(\frac{1}{1+e^{-a_i(\sum_{l=1}^{K}r_{li}-b_i)}}-d_i\Big)
\end{equation}
where $c_i = \frac{1+e^{a_ib_i}}{e^{a_ib_i}}$ and $d_i = \frac{1}{1+e^{a_ib_i}}$. So, it satisfies $U_i(0)=0$ and $U_i(\infty)=1$. We use the normalized logarithmic utility function, as in \cite{UtilityFairness}, that can be expressed as
\begin{equation}\label{eqn:log}
U_i(r_{1i}+r_{2i}+ ...+r_{Ki}) = \frac{\log(1+k_i\sum_{l=1}^{K}r_{li})}{\log(1+k_ir_{max})}
\end{equation}
where $r_{max}$ is the required rate for the user to achieve 100\% utility percentage and $k_i$ is the rate of increase of utility percentage with allocated rates. So, it satisfies $U_i(0)=0$ and $U_i(r_{max})=1$. %The logarithmic utility functions with $k=15$ and $k=0.1$ is shown in Figure \ref{fig:sigmoid}.
We consider the utility proportional fairness objective function that is given by
\begin{equation}\label{eqn:utility_fairness}
\underset{\textbf{r}}{\text{max}} \prod_{i=1}^{M}U_i(r_{1i} + r_{2i} + ... + r_{Ki})
\end{equation}
where $\textbf{r} =\{\textbf{r}_1, \textbf{r}_2,..., \textbf{r}_M\}$ and $\textbf{r}_i =\{r_{1i}, r_{2i},..., r_{Ki}\}$. The goal of this resource allocation objective function is to maximize the total system utility while ensuring proportional fairness between utilities (i.e., the product of the utilities of all UEs). This resource allocation objective function inherently guarantees:% utility proportional fairness which provides the following two characteristics:
\begin{itemize}
 \item non-zero resource allocation for all users. Therefore, the corresponding resource allocation optimization problem provides a minimum QoS for all users.
 \item  priority to users with real-time applications. Therefore, the corresponding resource allocation optimization problem improves the overall QoS for LTE system.
\end{itemize}

The basic formulation of the utility proportional fairness resource allocation problem is given by the following optimization problem:
\begin{equation}\label{eqn:opt_prob_fairness}
\begin{aligned}
& \underset{\textbf{r}}{\text{max}} & & \prod_{i=1}^{M}U_i(r_{1i} + r_{2i} + ... + r_{Ki}) \\
& \text{subject to} & & \sum_{i=1}^{M}r_{1i} \leq R_1, \sum_{i=1}^{M}r_{2i} \leq R_2, ...\\
& & & ...\:\:, \:\:\sum_{i=1}^{M}r_{Ki} \leq R_K,\\
& & & r_{li} \geq 0, \;\;\;l = 1,2, ...,K,\;\; i = 1,2, ...,M
\end{aligned}
\end{equation}
where $R_l$ is the total available rate at the $l^{th}$ carrier eNodeB.

We prove in Section \ref{sec:Proof} that the solution of the optimization problem (\ref{eqn:opt_prob_fairness}) is the global optimal solution.
%%%%%%%%%%%%%%%%%%%%%%%%%%%%%%%%%%%
%%%%%%%%%%%%%%%%%%%%%%%%%%%%%%%%%%%
%%%%%%%%%%%%%%%%%%%%%%%%%%%%%%%%%%%
%%%%%%%%%%%%%%%%%%%
%%%%%%%%%%%%%%%%%%%
\section{The Global Optimal Solution}\label{sec:Proof}

%There is a minimum user QoS that the primary carrier has to provide. Therefore, we assume that the primary carrier will allocate the resources based on utility proportional fairness.
In the optimization problem (\ref{eqn:opt_prob_fairness}), since the objective function $\arg \underset{\textbf{r}} \max \prod_{i=1}^{M}U_i(r_{1i}+r_{2i}+ ...+r_{Ki})$ is equivalent to $\arg \underset{\textbf{r}} \max \sum_{i=1}^{M}\log(U_i(r_{1i}+r_{2i}+ ...+r_{Ki}))$, then optimization problem (\ref{eqn:opt_prob_fairness}) can be written as:

\begin{equation}\label{eqn:opt_prob_fairness_mod}
\begin{aligned}
& \underset{\textbf{r}}{\text{max}} & & \sum_{i=1}^{M}\log \Big(U_i(r_{1i} + r_{2i} + ... + r_{Ki})\Big) \\
& \text{subject to} & & \sum_{i=1}^{M}r_{1i} \leq R_1, \sum_{i=1}^{M}r_{2i} \leq R_2, ...\\
& & & ...\:\:, \:\:\sum_{i=1}^{M}r_{Ki} \leq R_K,\\
& & & r_{li} \geq 0, \;\;\;l = 1,2, ...,K,\;\; i = 1,2, ...,M.
\end{aligned}
\end{equation}

\begin{lem}\label{lem:concavity}
The utility functions $\log(U_i(r_{1i} + ... + r_{Ki}))$ in the optimization problem (\ref{eqn:opt_prob_fairness_mod}) are strictly concave functions.
\end{lem}
\begin{proof}
In Section \ref{sec:Problem_formulation}, we assume that all the utility functions of the UEs are strictly concave or sigmoidal-like functions.

In the strictly concave utility function case, recall the utility function properties in Section \ref{sec:Problem_formulation}, the utility function is positive $ U_i(r_{1i} + ... + r_{Ki}) > 0$, increasing and twice differentiable with respect to $r_{li}$. Then, it follows that $\frac{ \partial U_i(r_{1i} + ... + r_{Ki})}{\partial r_{li}} > 0$ and $\frac{\partial^2 U_i(r_{1i}+ ... + r_{Ki})}{\partial r_{li}^2} < 0$. It follows that, the utility function $\log(U_i(r_{1i} + r_{2i} + ... + r_{Ki}))$ in the optimization problem (\ref{eqn:opt_prob_fairness_mod}) have
\begin{equation}\label{eqn:log_first_derivative}
\frac{\partial \log(U_i(r_{1i} + ... + r_{Ki}))}{\partial r_{li}} =  \frac{\frac{\partial U_i}{\partial r_{li}}}{U_i} > 0
\end{equation}
and
\begin{equation}\label{eqn:log_second_derivative}
\frac{\partial ^2\log(U_i(r_{1i} + ... + r_{Ki}))}{\partial r_{li}^2} =  \frac{\frac{\partial^2 U_i}{\partial r_{li}^2}U_i-(\frac{\partial U_i}{\partial r_{li}})^2}{U^2_i} < 0.
\end{equation}
Therefore, the strictly concave utility function $U_i(r_{1i} + r_{2i} + ... + r_{Ki})$ natural logarithm  $\log(U_i(r_{1i} + r_{2i} + ... + r_{Ki}))$ is also strictly concave. It follows that the natural logarithm  of the logarithmic utility function in equation (\ref{eqn:log}) is strictly concave.

In the sigmoidal-like utility function case, the utility function of the normalized sigmoidal-like function is given by equation (\ref{eqn:sigmoid}) as $U_i(r_{1i} + r_{2i} + ... + r_{Ki}) = c_i\Big(\frac{1}{1+e^{-a_i(\sum_{l=1}^{K}r_{li}-b_i)}}-d_i\Big)$. For $0<\sum_{l=1}^{K}r_{li}<\sum_{l=1}^{K}R_l$, we have
\begin{equation*}\label{eqn:sigmoid_bound}
\begin{aligned}
0&<c_i\Big(\frac{1}{1+e^{-a_i(\sum_{l=1}^{K}r_{li}-b_i)}}-d_i\Big)<1\\
d_i&<\frac{1}{1+e^{-a_i(\sum_{l=1}^{K}r_{li}-b_i)}}<\frac{1+c_id_i}{c_i}\\
\frac{1}{d_i}&>{1+e^{-a_i(\sum_{l=1}^{K}r_{li}-b_i)}}>\frac{c_i}{1+c_id_i}\\
0&<1-d_i({1+e^{-a_i(\sum_{l=1}^{K}r_{li}-b_i)}})<\frac{1}{1+c_id_i}\\
\end{aligned}
\end{equation*}
It follows that for $0<\sum_{l=1}^{K}r_{li}<\sum_{l=1}^{K}R_l$, we have the first and second derivative as
\begin{equation*}\label{eqn:sigmoid_derivative}
\begin{aligned}
\frac{\partial}{ \partial r_{li}}\log U_i(r_{1i} + ... + r_{Ki}) =& \frac{a_i d_i e^{-a_i(\sum_{l=1}^{K}r_{li}-b_i)}}{1-d_i(1+e^{-a_i(\sum_{l=1}^{K}r_{li}-b_i)})} \\
\;\;\;&  + \frac{a_ie^{-a_i(\sum_{l=1}^{K}r_{li}-b_i)}}{(1+e^{-a_i(\sum_{l=1}^{K}r_{li}-b_i)})}>0\\
\frac{\partial^2}{\partial r_{li}^2}\log U_i(r_{1i} + ... + r_{Ki}) =& \frac{-a_i^2d_ie^{-a_i(\sum_{l=1}^{K}r_{li}-b_i)}}{c_i\Big(1-d_i(1+e^{-a(\sum_{l=1}^{K}r_{li}-b_i)})\Big)^2} \\
\;\;\;&  + \frac{-a_i^2e^{-a_i(\sum_{l=1}^{K}r_{li}-b_i)}}{(1+e^{-a_i(\sum_{l=1}^{K}r_{li}-b_i)})^2} < 0\\
\end{aligned}
\end{equation*}
Therefore, the sigmoidal-like utility function $U_i(r_{1i}+...+r_{Ki})$ natural logarithm  $\log(U_i(r_{1i}+...+r_{Ki}))$ is strictly concave function. Therefore, all the utility functions in our model have strictly concave natural logarithm.
\end{proof}
%%%%%%%%%%%%%%%%%%%%%%%%%%%%%%%%%%%
\begin{thm}\label{thm:global_soln}
The optimization problem (\ref{eqn:opt_prob_fairness}) is a convex optimization problem and there exists a unique tractable global optimal solution.
\end{thm}
\begin{proof}
It follows from Lemma \ref{lem:concavity} that for all UEs utility functions are strictly concave. Therefore, the optimization problem (\ref{eqn:opt_prob_fairness_mod}) is a convex optimization problem \cite{Boyd2004}. The optimization problem (\ref{eqn:opt_prob_fairness_mod}) is equivalent to optimization problem (\ref{eqn:opt_prob_fairness}), therefore it is a convex optimization problem. For a convex optimization problem, there exists a unique tractable global optimal solution \cite{Boyd2004}.
\end{proof}

%%%%%%%%%%%%%%%%%%%%%%%%%%%%%%%%%%%
%%%%%%%%%%%%%%%%%%%%%%%%%%%%%%%%%%%
%%%%%%%%%%%%%%%%%%%%%%%%%%%%%%%%%%%
\section{The Dual Problem}\label{sec:Dual}

The key to a distributed and decentralized optimal solution of the primal problem in (\ref{eqn:opt_prob_fairness_mod}) is to convert it to the dual problem similar to \cite{Ahmed_Utility1}, \cite{kelly98ratecontrol} and \cite{Low99optimizationflow}. The optimization problem (\ref{eqn:opt_prob_fairness_mod}) can be divided into two simpler problems by using the dual problem.  We define the Lagrangian
\begin{equation}\label{eqn:lagrangian}
\begin{aligned}
L(\textbf{r},\textbf{p}) = & \sum_{i=1}^{M}\log \Big(U_i(r_{1i} + r_{2i} + ... + r_{Ki})\Big)\\
		   & -p_1(\sum_{i=1}^{M}r_{1i} + z_1 - R_1) - ...\\
		   & - p_K(\sum_{i=1}^{M}r_{Ki} + z_K - R_K)\\
                = &  \sum_{i=1}^{M}\Big({\log(U_i(r_{1i} + r_{2i} + ... + r_{Ki}))-\sum_{l=1}^{K}p_lr_{li}\Big)}\\
                  & + \sum_{l=1}^{K} p_l(R_l-z_l)\\
                = &  \sum_{i=1}^{M}L_i(\textbf{r}_i,\textbf{p}) + \sum_{l=1}^{K} p_l(R_l-z_l)\\
\end{aligned}
\end{equation}
where $z_l\geq 0$ is the $l^{th}$ slack variable and $p_l$ is Lagrange multiplier or the shadow price of the $l^{th}$ carrier eNodeB (i.e. the total price per unit rate for all the users in the coverage area of the $l^{th}$ carrier eNodeB) and $\textbf{p}=\{p_1,p_2,...,p_K\}$. Therefore, the $i^{th}$ UE bid for rate from the $l^{th}$ carrier eNodeB can be written as $w_{li} = p_l r_{li}$ and we have $\sum_{i=1}^{M}w_{li} = p_l \sum_{i=1}^{M}r_{li}$. The first term in equation (\ref{eqn:lagrangian}) is separable in $\textbf{r}_i$. So we have $\underset{\textbf{r}}\max \sum_{i=1}^{M}({\log(U_i(r_{1i}+r_{2i}+...+r_{Ki}))-\sum_{l=1}^{K}p_lr_{li})} = \sum_{i=1}^{M}\underset{{\textbf{r}_i}}\max\big({\log(U_i(r_{1i}+r_{2i}+...+r_{Ki}))-\sum_{l=1}^{K}p_lr_{li}\big)}$.
The dual problem objective function can be written as
\begin{equation}\label{eqn:dual_obj_fn}
\begin{aligned}
D(\textbf{p}) = & \underset{{\textbf{r}}}\max \:L(\textbf{r},\textbf{p}) \\
= &\sum_{i=1}^{M}\underset{{\textbf{r}_i}}\max (L_i(\textbf{r}_i,\textbf{p})) + \sum_{l=1}^{K} p_l(R_l-z_l)
\end{aligned}
\end{equation}
The dual problem is given by
\begin{equation}\label{eqn:dual_problem}
\begin{aligned}
& \underset{{\textbf{p}}}{\text{min}}
& & D(\textbf{p}) \\
& \text{subject to}
& & p_l \geq 0, \;\;\;\;\;l = 1,2, ...,K.
\end{aligned}
\end{equation}
So we have
\begin{equation}\label{eqn:dual_max}
\frac{\partial D(\textbf{p})}{\partial p_l} =  R_l-\sum_{i=1}^{M}r_{li} -z_l = 0
\end{equation}
substituting by $\sum_{i=1}^{M}w_{li} = p_l \sum_{i=1}^{M}r_{li}$ we have
\begin{equation}\label{eqn:dual_new_obj}
p_l = \frac{\sum_{i=1}^{M}w_{li}}{R_l-z_l}.
\end{equation}
Now, we divide the primal problem (\ref{eqn:opt_prob_fairness_mod}) into two simpler optimization problems in the UEs and the eNodeBs. The $i^{th}$ UE optimization problem is given by:
\begin{equation}\label{eqn:opt_prob_fairness_UE}
\begin{aligned}
& \underset{{r_i}}{\text{max}}
& & \log(U_i(r_{1i} + r_{2i} + ... + r_{Ki}))-\sum_{l=1}^{K}p_lr_{li} \\
& \text{subject to}
& & p_l \geq 0\\
& & &  r_{li} \geq 0, \;\;\;\;\; i = 1,2, ...,M, l = 1,2, ...,K.
\end{aligned}
\end{equation}

The second problem is the $l^{th}$ eNodeB optimization problem for rate proportional fairness that is given by:
\begin{equation}\label{eqn:opt_prob_fairness_eNodeB}
\begin{aligned}
& \underset{p_l}{\text{min}}
& & D(\textbf{p}) \\
& \text{subject to}
& & p_l \geq 0.\\
\end{aligned}
\end{equation}
The minimization of shadow price $p_l$ is achieved by the minimization of the slack variable $z_l \geq 0$ from equation (\ref{eqn:dual_new_obj}). Therefore, the maximum utility percentage of the $l^{th}$ eNodeB rate $R_l$ is achieved by setting the slack variable $z_l = 0$. In this case, we replace the inequality in primal problem (\ref{eqn:opt_prob_fairness_mod}) constraints by  equality constraints and so we have $\sum_{i=1}^{M}w_{li} = p_l R_l$. Therefore, we have $p_l = \frac{\sum_{i=1}^{M}w_{li}}{R_l}$ where $w_{li} = p_l r_{li}$ is transmitted by the $i^{th}$ UE to $l^{th}$ eNodeB. The utility proportional fairness in the objective function of the optimization problem (\ref{eqn:opt_prob_fairness}) is guaranteed in the solution of the optimization problems (\ref{eqn:opt_prob_fairness_UE}) and (\ref{eqn:opt_prob_fairness_eNodeB}).

\section{Distributed Optimization Algorithm}\label{sec:Algorithm}

The distributed resource allocation algorithm, in \cite{Ahmed_Utility4}, for optimization problems (\ref{eqn:opt_prob_fairness_UE}) and (\ref{eqn:opt_prob_fairness_eNodeB}) is a modified version of the distributed algorithms in \cite{Ahmed_Utility1, Ahmed_Utility2,Ahmed_Utility3}, \cite{kelly98ratecontrol} and \cite{Low99optimizationflow}, which is an iterative solution for allocating the network resources for a single carrier. The algorithm in \cite{Ahmed_Utility4} allocates resources from multiple carriers simultaneously with utility proportional fairness policy. The algorithm is divided into the $i^{th}$ UE algorithm as shown in Algorithm 1 \cite{Ahmed_Utility4} and the $l^{th}$ eNodeB carrier algorithm as shown in Algorithm 2 \cite{Ahmed_Utility4}. In Algorithm 1 and 2 \cite{Ahmed_Utility4}, the $i^{th}$ UE starts with an initial bid $w_{li}(1)$ which is transmitted to the $l^{th}$ carrier eNodeB. The $l^{th}$ eNodeB calculates the difference between the received bid $w_{li}(n)$ and the previously
received bid $w_{li}(n-1)$ and exits if it is less than a pre-specified threshold $\delta$. We set $w_{li}(0) = 0$. If the value is greater than the threshold, the $l^{th}$ eNodeB calculates the shadow price $p_l(n) = \frac{\sum_{i=1}^{M}w_{li}(n)}{R_l}$ and sends that value to all UEs in its coverage area. The $i^{th}$ UE receives the shadow prices $p_{l}$ from all in range carriers eNodeBs and compares them to find the first minimum shadow price $p_{\min}^{1}(n)$ and the corresponding carrier index $l_1 \in L$ where $L = \{1, 2, ..., K\}$. The $i^{th}$ UE solves for the $l_1$ carrier rate $r_{l_1i}(n)$ that maximizes $\log U_i(r_{1i}+...+r_{Ki}) - \sum_{l=1}^{K}p_l(n)r_{li}$ with respect to $r_{l_1i}$. The rate $r_{i}^{1}(n) = r_{l_1i}(n)$ is used to calculate the new bid $w_{l_1i}(n)=p_{\min}^{1}(n) r_{i}^{1}(n)$. The $i^{th}$ UE sends the value of its new bid $w_{l_1i}(n)$ to the $l_1$ carrier eNodeB. Then, the $i^{th}$ UE selects the second minimum shadow price  $p_{\min}^{2}(n)$ and the corresponding
carrier index $l_2 \in L$. The $i^{th}$ UE solves
for the $l_2$ carrier rate $r_{l_2i}(n)$ that maximizes $\log U_i(r_{1i}+...+r_{Ki}) - \sum_{l=1}^{K}p_l(n)r_{li}$ with respect to $r_{l_2i}$. The rate $r_{l_2i}(n)$ subtracted by the rate from $l_1$ carrier $r_{i}^{2}(n) = r_{l_2i}(n) - r_{i}^{1}(n)$ is used to calculate the new bid $w_{l_2i}(n)=p_{\min}^{2}(n) r_{i}^{2}(n)$ which is sent to $l_2$ carrier eNodeB. In general, the $i^{th}$ UE selects the $m^{th}$  minimum shadow price $p_{\min}^{m}(n)$ with carrier index $l_m \in L$ and solves for the $l_m$ carrier rate $r_{l_mi}(n)$ that maximizes $\log U_i(r_{1i}+...+r_{Ki}) - \sum_{l=1}^{K}p_l(n)r_{li}$ with respect to $r_{l_mi}$. The rate $r_{l_mi}(n)$ subtracted by $l_1, l_2, ..., l_{m-1}$ carriers rates  $r_{i}^{m}(n) = r_{l_mi}(n) - (r_{i}^{1}(n)+r_{i}^{2}(n)+...+r_{i}^{m-1}(n))$ is used to calculate the new bid $w_{l_mi}(n)=p_{\min}^{m}(n) r_{i}^{m}(n)$ which is sent to $l_m$ carrier eNodeB. This process is repeated until $|w_{li}(n) -w_{li}(n-1)|$ is less than the threshold $\delta$ for all $l$
carriers.

The distributed algorithm in \cite{Ahmed_Utility4} is set to avoid the situation of allocating zero rate to any user (i.e. no user is dropped). This is inherited from the utility proportional fairness policy in the optimization problem, similar to \cite{Ahmed_Utility1}, \cite{Ahmed_Utility2} and \cite{Ahmed_Utility3}. In addition, the UE chooses from the nearby carriers eNodeBs the one with the lowest shadow price and starts requesting bandwidth from that carrier eNodeB. If the allocated rate is not enough or the price of the resources increases due to high demand on that carrier eNodeB resources from other UEs, the UE switches to another nearby eNodeB carrier with a lower resource price to be allocated the rest of the required resources. This is done iteratively until an equilibrium between demand and supply of resources is achieved and the optimal rates are allocated in the LTE mobile network. Figure \ref{fig:multiple_app_flow_centralized} shows a block diagram that represents the distributed RA algorithm.
%%%%%%%%%%%%%%%%%%%%%%%%%%
\begin{figure}[t!]
\centering
  \includegraphics[width=0.8\plotwidth]{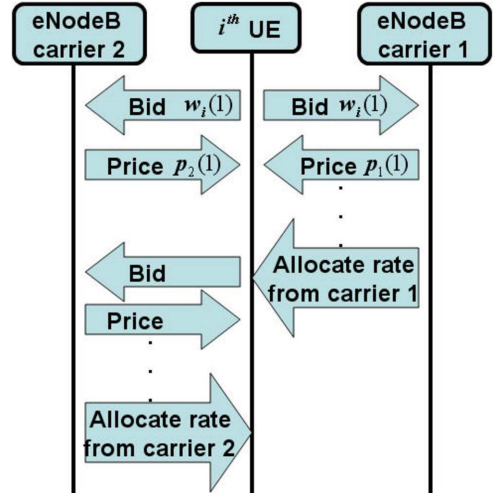}
  \caption{Flow Diagram with the assumption that the shadow price from the first carrier eNodeB $p_1$ is less before the $n_1$th iteration so rate $r_{1i}$ of the $i^{th}$ user is allocated. After the $n_1$th iteration, the shadow price from the second carrier eNodeB $p_2$ is less so rate $r_{2i}$ is allocated.}
  \label{fig:multiple_app_flow_centralized}
\end{figure}
%%%%%%%%%%%%%%%%%%%%%%%%%%
\section{Convergence Analysis}\label{sec:conv_analy}
In this section, we present the convergence analysis of Algorithm 1 and 2 in \cite{Ahmed_Utility4} for different values of carriers eNodeBs rates $R_l$. This analysis is equivalent to low and high-traffic hours analysis in cellular systems (e.g. change in the number of active users $M$ and their traffic in the cellular system \cite{Ahmed_Utility2}).
\subsection{Drawback in Algorithm 1 and 2 in \cite{Ahmed_Utility4}}\label{sec:conv_drawbacks}
%%%%%%%%%%%%%%%%%%%%%%%%%%%%%%%%%%%%%
%%%%%%%%%%%%%%%%%%%%%%%%%%%%%%%%%%%%%
\begin{lem}\label{lem:slope_curve}
For sigmoidal-like utility function $U_i(r_{1i}+r_{2i}+ ...+r_{Ki})$,  the slope curvature function $\frac{\partial \log U_i(r_{1i}+r_{2i}+ ...+r_{Ki})}{\partial r_{li}}$ has an inflection point at $\sum_{l=1}^{K}r_{li} = r_i^{s} \approx b_i$ and is convex for $\sum_{l=1}^{K}r_{li} > r_i^{s}$.
\end{lem}
\begin{proof}
For the sigmoidal-like function $U_i(r_{1i}+r_{2i}+ ...+r_{Ki}) = c_i\Big(\frac{1}{1+e^{-a_i(\sum_{l=1}^{K}r_{li}-b_i)}}-d_i\Big)$, let $S_i(r_{li}) = \frac{\partial \log U_i(r_{1i}+r_{2i}+ ...+r_{Ki})}{\partial r_{li}}$ be the slope curvature function. Then, we have that
\begin{equation}
\begin{aligned}\label{eqn:diff_slope}
\frac{\partial S_i}{\partial r_{li}} &= \frac{-a_i^2d_ie^{-a_i(\sum_{l=1}^{K}r_{li}-b_i)}}{c_i\Big(1-d_i(1+e^{-a_i(\sum_{l=1}^{K}r_{li}-b_i)})\Big)^2} \\
&- \frac{a_i^2e^{-a_i(\sum_{l=1}^{K}r_{li}-b_i)}}{\Big(1+e^{-a_i(\sum_{l=1}^{K}r_{li}-b_i)}\Big)^2}\\
\text{and}\\
\frac{\partial^2 S_i}{\partial r_{li}^2}& = \frac{a_i^3d_ie^{-a_i(\sum_{l=1}^{K}r_{li}-b_i)}(1-d_i(1-e^{-a_i(\sum_{l=1}^{K}r_{li}-b_i)}))}{c_i\Big(1-d_i(1+e^{-a_i(\sum_{l=1}^{K}r_{li}-b_i)})\Big)^3} \\
+& \frac{a_i^3e^{-a_i(\sum_{l=1}^{K}r_{li}-b_i)}(1-e^{-a_i(\sum_{l=1}^{K}r_{li}-b_i)})}{\Big(1+e^{-a_i(\sum_{l=1}^{K}r_{li}-b_i)}\Big)^3}.\\
\end{aligned}
\end{equation}
We analyze the curvature of the slope of the natural logarithm of sigmoidal-like utility function. For the first derivative, we have $\frac{\partial S_i}{\partial r_{li}}<0 \:\:\:\forall\: r_{li}$. The first term $S^1_i$ of $\frac{\partial^2 S_i}{\partial r_{li}^2}$ in equation (\ref{eqn:diff_slope}) can be written as
\begin{equation}\label{eqn:slope_fn}
S^1_i = \frac{a_i^3e^{a_ib_i}(e^{a_ib_i}+e^{-a_i(\sum_{l=1}^{K}r_{li}-b_i)})}{(e^{a_ib_i}-e^{-a_i(\sum_{l=1}^{K}r_{li}-b_i)})^3}
\end{equation}
and we have the following properties:
\begin{equation}\label{eqn:slope_fn_term1}
\left\{
\begin{array}{l l}
   \lim_{\sum_{l=1}^{K}r_{li} \rightarrow 0} S^1_i = \infty,\\
   \lim_{\sum_{l=1}^{K}r_{li} \rightarrow b_i} S^1_i = 0 \:\:\text{for} \:\:b_i 	\gg \frac{1}{a_i}.
\end{array} \right.
\end{equation}
For second term $S^2_i$ of $\frac{\partial^2 S_i}{\partial r_i^2}$ in equation (\ref{eqn:diff_slope}), we have the following properties:
\begin{equation}\label{eqn:slope_fn_term2}
\left\{
\begin{array}{l l}
   S^2_i (r_{li}=b_{i}-\sum_{j\neq l} r_{ji}) = 0,\\
   S^2_i (r_{li}>b_{i}-\sum_{j\neq l} r_{ji}) > 0,\\
   S^2_i (r_{li}<b_{i}-\sum_{j\neq l} r_{ji}) < 0.
\end{array} \right.
\end{equation}
From equation (\ref{eqn:slope_fn_term1}) and (\ref{eqn:slope_fn_term2}), $S_i$ has an inflection point at $\sum_{l=1}^{K}r_{li} = r_i^{s}\approx b_i$. In addition, we have the curvature of $S_i$ changes from a convex function close to origin to a concave function before the inflection point $\sum_{l=1}^{K}r_{li} = r_i^{s}$ then to a convex function after the inflection point.
\end{proof}
%%%%%%%%%%%%%%%%%%%%%%%%%%%%%%%%%%%%%
Our rate allocation approach guarantees non-zero rate allocation for all active users in the coverage area of a specific carrier eNodeB. We define the set $\mathcal{M}^{l}:=\{i:r_{li} \neq 0\}$ to be the set of active users covered by the $l^{th}$ eNodeB. Then, we have the following Corollary.
%%%%%%%%%%%%%%%%%%%%%%%%%%%%%%%%%%%%%
\begin{cor}\label{cor:sig_convergence}
If $\sum_{i \in \mathcal{M}^{l}}r_i^{\text{inf}} \ll R_l \: \forall \: l\in L$ then Algorithm 1 and  2 in \cite{Ahmed_Utility4} converge to the global optimal rates which correspond to the steady state shadow price $p_{ss}< \frac{a_{i_{\max}} d_{i_{\max}} }{1-d_{i_{\max}} }+\frac{a_{i_{\max} }}{2}$ where $i_{\max} = \arg \max_{i \in \mathcal{M}^{l}} b_i$.
\end{cor}
\begin{proof}
For the sigmoidal-like function $U_i(r_{1i}+r_{2i}+ ...+r_{Ki}) = c_i\Big(\frac{1}{1+e^{-a_i(\sum_{l=1}^{K}r_{li}-b_i)}}-d_i\Big)$, the optimal solution is achieved by solving the optimization problem (\ref{eqn:opt_prob_fairness_mod}). In Algorithm 1 \cite{Ahmed_Utility4}, an important step to reach to the optimal solution is to solve the optimization problem $r_{li}(n) = \arg \underset{r_{li}}\max \Big(\log U_i(r_{1i}+r_{2i}+ ...+r_{Ki}) - p_l(n)r_{li}\Big)$ for every UE in the $l^{th}$ eNodeB coverage area. The solution of this problem can be written, using Lagrange multipliers method, in the form
\begin{equation}\label{eqn:slope_equation}
\frac{\partial \log U_i(r_{1i}+r_{2i}+ ...+r_{Ki})}{\partial r_{li}}-p_l =  S_i(r_{li}) - p_l = 0.
\end{equation}
From equation (\ref{eqn:slope_fn_term1}) and (\ref{eqn:slope_fn_term2}) in Lemma \ref{lem:slope_curve}, we have the curvature of $S_i(r_{li})$ is convex for $\sum_{l=1}^{K}r_{li}  > r_i^s \approx b_i$. The algorithm in \cite{Ahmed_Utility4} is guaranteed to converge to the global optimal solution when the slope $S_i(r_{li})$ of all the utility functions natural logarithm $\log U_i(r_{1i}+r_{2i}+ ...+r_{Ki})$ are in the convex region of the functions, similar to analysis of logarithmic functions in \cite{kelly98ratecontrol} and \cite{Low99optimizationflow}. Therefore, the natural logarithm of sigmoidal-like functions $\log U_i(r_{1i}+r_{2i}+ ...+r_{Ki})$ converge to the global optimal solution for $\sum_{l=1}^{K}r_{li} > r_i^s \approx b_i$. The inflection point of sigmoidal-like function $U_i(r_{1i}+r_{2i}+ ...+r_{Ki})$ is at $r_i^{\text{inf}} = b_i$. For $\sum_{i \in \mathcal{M}^{l}}r_i^{\text{inf}} \ll R_l$, the algorithm in \cite{Ahmed_Utility4} allocates rates $\sum_{l=1}^{K}r_{li}>b_i$ for all users.
Since $S_i(r_{li})$ is convex for $\sum_{l=1}^{K}r_{li}>r_i^s \approx b_i$ then the optimal solution can be achieved by Algorithm 1 and 2 in \cite{Ahmed_Utility4}. We have from equation (\ref{eqn:slope_equation}) and as $S_i(r_{li})$ is convex for $\sum_{l=1}^{K}r_{li} > r_i^s \approx b_i$, that $p_{ss}< S_i(\sum_{l=1}^{K}r_{li} =\max_{i \in \mathcal{M}^{l}} b_i)$ where $S_i(\sum_{l=1}^{K}r_{li} =\max_{i \in \mathcal{M}^{l}} b_i) = \frac{a_{i_{\max}} d_{i_{\max}} }{1-d_{i_{\max}} }+\frac{a_{i_{\max} }}{2}$ and $i_{\max} = \arg \max_{i \in \mathcal{M}^{l}} b_i$.
\end{proof}
%%%%%%%%%%%%%%%%%%%%%%%%%%%%%%%%%%%%%
%%%%%%%%%%%%%%%%%%%%%%%%%%%%%%%%%%%%%
We define the set $\mathcal{M}^{\mathcal{L}}:=\{i:r_{li} \neq 0 \: \forall \: l\in \mathcal{L}, r_{li} = 0 \: \forall \: l\notin \mathcal{L} \}$ to be the set of active users covered exclusively by the set of carriers eNodeBs $\mathcal{L} \subseteq L$. Then, we have the following Corollary.
%%%%%%%%%%%%%%%%%%%%%%%%%%%%%%%%%%%%%
%%%%%%%%%%%%%%%%%%%%%%%%%%%%%%%%%%%%%
\begin{cor}\label{cor:sig_fluctuate}
For $\sum_{i \in \mathcal{M}^{\mathcal{L}}}r_i^{\text{inf}}> \sum_{l \in \mathcal{L}}R_l$ and the global optimal shadow price $p_{ss} \approx \frac{a_id_i e^{\frac{a_ib_i}{2}}}{1-d_i(1+e^{\frac{a_ib_i}{2}})} + \frac{a_ie^{\frac{a_ib_i}{2}}}{(1+e^{\frac{a_ib_i}{2}})}$ where $i \in \mathcal{M}^{\mathcal{L}}$, then the solution given by Algorithm 1 and 2 in \cite{Ahmed_Utility4} fluctuates about the global optimal rates.
\end{cor}
\begin{proof}
For the sigmoidal-like function $U_i(r_{1i}+r_{2i}+ ...+r_{Ki}) = c_i\Big(\frac{1}{1+e^{-a_i(\sum_{l=1}^{K}r_{li}-b_i)}}-d_i\Big)$, it follows from lemma \ref{lem:slope_curve} that for  $\sum_{i \in \mathcal{M}^{\mathcal{L}}}r_i^{\text{inf}}> \sum_{l \in \mathcal{L}}R_l$  $\exists \:\: {i \in \mathcal{M}^{\mathcal{L}}}$ such that the optimal rates $\sum_{l=1}^{K}r_{li}^{\text{opt}} < b_i$. Therefore, if $p_{ss} \approx \frac{a_id_i e^{\frac{a_ib_i}{2}}}{1-d_i(1+e^{\frac{a_ib_i}{2}})} + \frac{a_ie^{\frac{a_ib_i}{2}}}{(1+e^{\frac{a_ib_i}{2}})}$ is the optimal shadow price for optimization problem (\ref{eqn:opt_prob_fairness_mod}). Then, a small change in the shadow price $p_l(n)$ in the $n^{th}$ iteration can lead the rate $r_{li}(n)$ (root of $S_i(r_{li}) - p_l(n) =0$) to fluctuate between the concave and convex curvature of the slope curve $S_i(r_{li})$ for the $i^{th}$ user. Therefore, it causes fluctuation in the bid $w_{li}(n)$ sent to the eNodeB and fluctuation in the shadow price $p_l(n)$ set by eNodeB.
Therefore, the iterative solution of Algorithm 1 and 2 in \cite{Ahmed_Utility4} fluctuates about the global optimal rates $\sum_{l=1}^{K}r_{li}^{\text{opt}}$.
\end{proof}
%%%%%%%%%%%%%%%%%%%%%%%%%%%%%%%%%%%%%
%%%%%%%%%%%%%%%%%%%%%%%%%%%%%%%%%%%%%
\begin{thm}\label{thm:sig_not_conv}
Algorithm 1 and 2 in \cite{Ahmed_Utility4} does not converge to the global optimal rates for all values of $R_l$.
\end{thm}
\begin{proof}
It follows from Corollary \ref{cor:sig_convergence} and \ref{cor:sig_fluctuate} that Algorithm 1 and 2 in \cite{Ahmed_Utility4} does not converge to the global optimal rates for all values of $R_l$.
\end{proof}
%%%%%%%%%%%%%%%%%%%%%%%%%%%%%%%%%%%%%%%%%%%%%%%%%
%%%%%%%%%%%%%%%%%%%%%%%%%%%%%%%%%%%%%%%%%%%%%%%%%
%%%%UE pseudocode

\begin{algorithm}[tb]
\caption{The $i^{th}$ UE Algorithm}\label{alg:UE_FK}
\begin{algorithmic}
\STATE {Send initial bid $w_{li}(1)$ to $l^{th}$ carrier eNodeB (where $l \in L = \{1, 2, ..., K\}$)}
\LOOP
	\STATE {Receive shadow prices $p_{l\in L}(n)$ from all in range carriers eNodeBs}
	\IF {STOP from all in range carriers eNodeBs} %
	   	%	\STATE {STOP and confirm to eNodeB}
		\STATE {Calculate allocated rates $r_{li} ^{\text{opt}}=\frac{w_{li}(n)}{p_l(n)}$}
		\STATE {STOP}
	\ELSE
		\STATE{Set $p_{\min}^{0} = \{\}$ and $r_{i}^{0}=0$}
		\FOR{$m = 1 \to K$}
		    \STATE{$p_{\min}^{m}(n) = \min (\textbf{p} \setminus \{p_{\min}^{0},p_{\min}^{1},...,p_{\min}^{m-1}\})$}
		    \STATE{$l_m = \{l \in L : p_{l}=\min (\textbf{p} \setminus \{p_{\min}^{0}, p_{\min}^{1}, ..., p_{\min}^{m-1}\}) \}$}
		    \COMMENT{$l_m$ is the index of the corresponding carrier}
		    \STATE {Solve $r_{l_mi}(n) = \arg \underset{r_{l_mi}}\max \Big(\log U_i(r_{1i}+ ... + r_{Ki}) - \sum_{l=1}^{K}p_l(n)r_{li}\Big)$ for the $l_m$ carrier eNodeB}
		    \STATE {$r_{i}^{m}(n) = r_{l_mi}(n)-\sum_{j=0}^{m-1}r_{i}^{j}(n)$}
		    \IF{$r_{i}^{m}(n)<0$}
			      \STATE{Set $r_{i}^{m}(n)=0$}
		    \ENDIF
		    %\STATE {Send new bid $w_{l_mi} (n)= p_{\min}^{m}(n) r_{i}^{m}(n)$ to $l_m$ carrier eNodeB}
		    \STATE {Calculate new bid $w_{l_mi} (n)= p_{\min}^{m}(n) r_{i}^{m}(n)$}
            %\STATE{$\star \star \star$ \{Start addition: Algorithm (\ref{alg:UE_FK}$^{\star}$)\}}
		    \IF {$|w_{l_mi}(n) -w_{l_mi}(n-1)| >\Delta w(n)$} %
			    \STATE {$w_{l_mi}(n) =w_i(n-1) + \text{sign}(w_{l_mi}(n) -w_{l_mi}(n-1))\Delta w(n)$}
			    \COMMENT {$\Delta w = h_1 e^{-\frac{n}{h_2}}$ or $\Delta w = \frac{h_3}{n}$}
		    \ENDIF
            %\STATE{$\star \star \star$ \{End addition: Algorithm (\ref{alg:UE_FK}$^{\star}$)\}}
		    \STATE {Send new bid $w_{l_mi} (n)$ to $l_m$ carrier eNodeB}
		\ENDFOR
	\ENDIF
\ENDLOOP
\end{algorithmic}
\end{algorithm}
%%%%%%%%%%%%%%%%%%%%%%%%%%%%%%%%%%%%%%%%%%%%%%%%%%
\subsection{Solution using Algorithm \ref{alg:UE_FK} and \ref{alg:eNodeB_FK}}\label{sec:conv_solution}
For a robust algorithm, we add a fluctuation decay function to the algorithm presented in \cite{Ahmed_Utility4} as shown in Algorithm \ref{alg:UE_FK}. Our robust algorithm ensures convergence for all values of the carriers eNodeBs maximum rate $R_l$ for all $l$. Algorithm \ref{alg:UE_FK} and \ref{alg:eNodeB_FK} allocated rates coincide with Algorithm 1 and 2 in \cite{Ahmed_Utility4} for $\sum_{i \in \mathcal{M}^{l}}r_i^{\text{inf}} \ll R_l \:\: \forall \:\: l \in L$. For $\sum_{i \in \mathcal{M}^{\mathcal{L}}}r_i^{\text{inf}}> \sum_{l \in \mathcal{L}}R_l$, robust algorithm avoids the fluctuation in the non-convergent region discussed in the previous section. This is achieved by adding a convergence measure $\Delta w(n)$ that senses the fluctuation in the bids $w_{li}$. In case of fluctuation, it decreases the step size between the current and the previous bid $w_{li}(n) -w_{li}(n-1)$ for every user $i$ using \textit{fluctuation decay function}. The
fluctuation decay function could be in the following forms:
\begin{itemize}
\item \textit{Exponential function}: It takes the form $\Delta w(n) = h_1 e^{-\frac{n}{h_2}}$.
\item \textit{Rational function}: It takes the form $\Delta w(n) = \frac{h_3}{n}$.
\end{itemize}
where $h_1, h_2, h_3$ can be adjusted to change the rate of decay of the bids $w_{li}$. %The robust algorithm with the fluctuation decay function is in Algorithm (\ref{alg:UE_FK}$^{\star}$) and (\ref{alg:eNodeB_FK}).

\begin{rem}
The fluctuation decay function can be included in the UE or the eNodeB Algorithm.
\end{rem}
In our model, we add the decay part to the UE Algorithm as shown in Algorithm \ref{alg:UE_FK}.
%%%%%%%%%%%%%%%%%%%%%%%%%%%%%%%%%%%
%%%%%%%eNodeB pseudocode

\begin{algorithm}[tb]
\caption{The $l^{th}$ eNodeB Algorithm}\label{alg:eNodeB_FK}
\begin{algorithmic}
\LOOP
	\STATE {Receive bids $w_{li}(n)$ from UEs}
	\COMMENT{Let $w_{li}(0) = 0\:\:\forall i$}
			\IF {$|w_{li}(n) -w_{li}(n-1)|<\delta  \:\:\forall i$} %
	   		\STATE {Allocate rates, $r_{li}^{\text{opt}}=\frac{w_{li}(n)}{p_l(n)}$ to $i^{th}$ UE}
	   		\STATE {STOP}
		\ELSE
	\STATE {Calculate $p_l(n) = \frac{\sum_{i=1}^{M}w_{li}(n)}{R_l}$}
	\STATE {Send new shadow price $p_l(n)$ to all UEs}
	%\IF {STOP confirmation from the $i^{th}$ UE} %
	   %	\STATE {Allocate rate $r_{i,opt}$ to user $i$}
	\ENDIF
	%\COMMENT{Do we need an exit criteria or this is endless process?}
\ENDLOOP
\end{algorithmic}
\end{algorithm}

%%%%%%%%%%%%%%%%%%%%%%%%%%%%%%%%%%%
%%%%%%%%%%%%%%%%%%%%%%%%%%%%%%%%%%%
%%%%%%%%%%%%%%%%%%%%%%%%%%%%%%%%%%%
\section{Simulation Results}\label{sec:sim}

Algorithm \ref{alg:UE_FK} and \ref{alg:eNodeB_FK} were applied to various logarithmic and sigmoidal-like utility functions with different parameters in MATLAB. The simulation results showed convergence to the global optimal rates. In this section, we present the simulation results for two carriers in a heterogeneous network (HetNet) that consists of one macro cell, one small cell and $12$ active UEs as shown in Figure \ref{fig:RA_system_model}. The UEs are divided into two groups. The $1^{st}$ group of UEs (index $i=\{1,2,3,4,5,6\}$) is located in the macro cell under the coverage area of both the $1^{st}$ carrier (C1) and the $2^{nd}$ carrier (C2) eNodeBs. We use three normalized sigmoidal-like functions that are expressed by equation (\ref{eqn:sigmoid}) with different parameters. The used parameters are $a = 5$, $b=10$ corresponding to a sigmoidal-like function that is an approximation to a step function at rate $r =10$ (e.g. VoIP) and is the utility of UEs with indexes $i=\{1,7\}$, $a = 3$, $b=20$
corresponding to a sigmoidal-like function that is an approximation of an adaptive real-time application with inflection point at rate $r=20$ (e.g. standard definition video streaming) and is the utility of UEs with indexes $i=\{2,8\}$, and $a = 1$,  $b=30$ corresponding to a sigmoidal-like function that is also an approximation of an adaptive real-time application with inflection point at rate $r=30$ (e.g. high definition video streaming) and is the utility of UEs with indexes $i=\{3,9\}$, as shown in Figure \ref{fig:utility}. We use three logarithmic functions that are expressed by equation (\ref{eqn:log}) with $r_{max} =100$ and different $k_i$ parameters which are approximations for delay-tolerant applications (e.g. FTP). We use $k =15$ for UEs with indexes $i=\{4,10\}$, $k =3$ for UEs with indexes $i=\{5,11\}$, and $k = 0.5$ for UEs with indexes $i=\{6,12\}$, as shown in Figure \ref{fig:utility}. A summary is shown in table \ref{table:parameters}. A three dimensional view of the sigmoidal-like utility
function $U_i(r_{1i} +r_{2i})$ is show in Figure \ref{fig:Utility3D}.
%%%%%%%%%%%%%%%%%%%%%%%%%%
\begin{figure}[tb]
\centering
\includegraphics[width=3.5in]{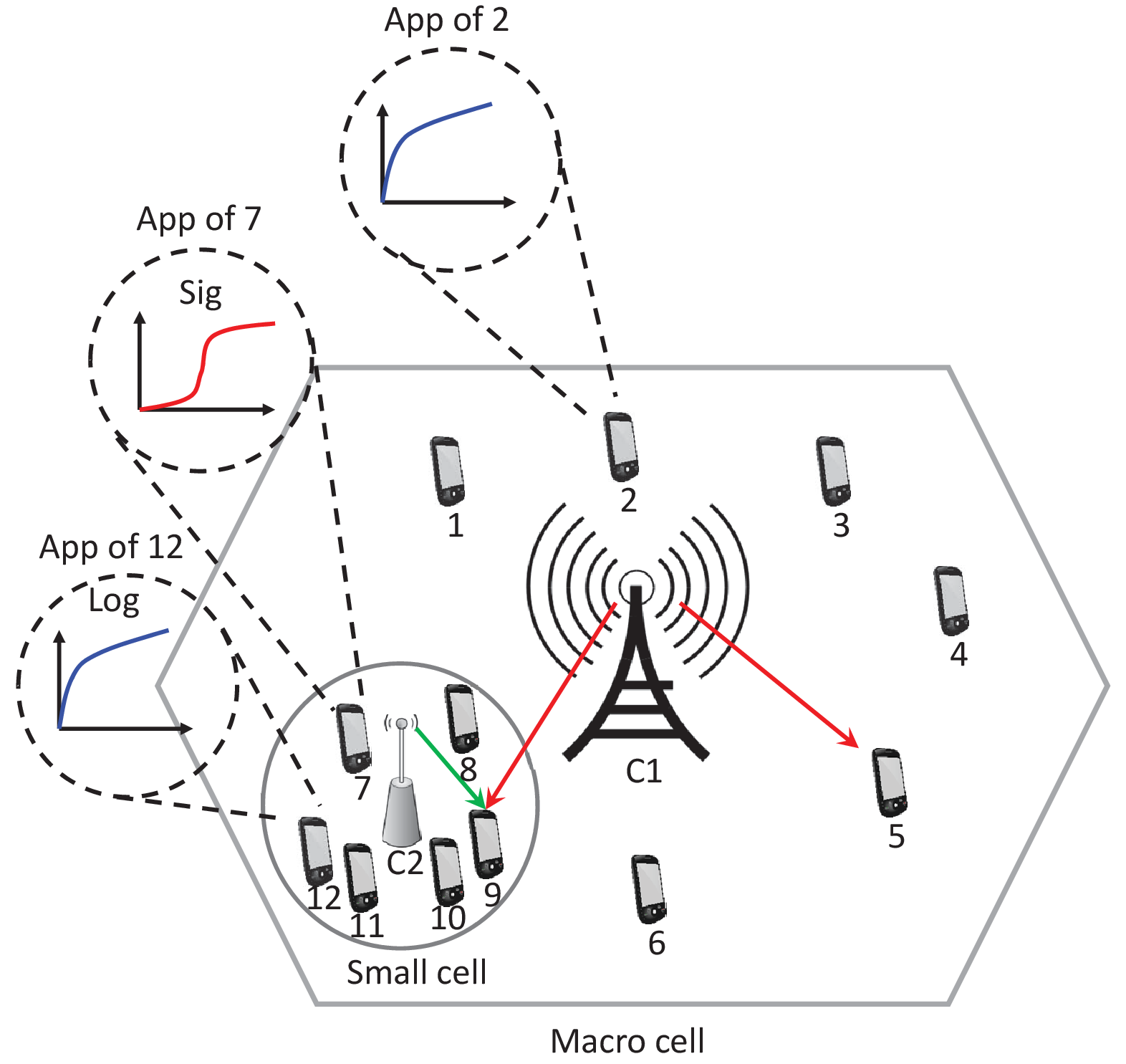}
\caption{System model with two groups of users. The $1^{st}$ group with UE indexes $i =\{ 1,2,3,4,5,6\}$, $2^{nd}$ group with UE indexes $i = \{7,8,9,10,11,12\}$.}
%%\myfigureshrinker{\vspace{-0.06in}}
\label{fig:RA_system_model}
\end{figure}
%%%%%%%%%%%%%%%%%%%%%%%%%
\begin{figure}[tb]
\centering
\includegraphics[width=3.5in]{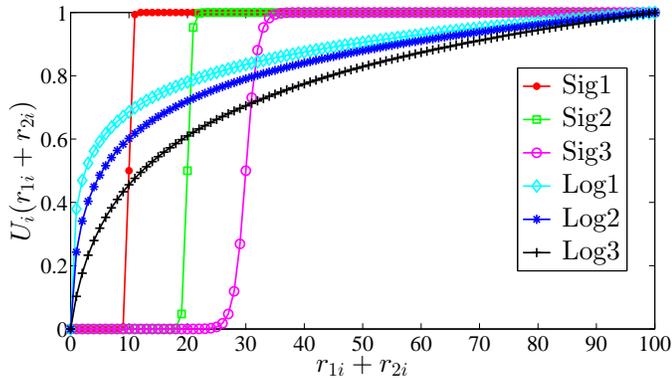}
\caption{The users utility functions $U_i(r_{1i}+r_{2i})$ used in the simulation (three sigmoidal-like functions and three logarithmic functions).}
%%\myfigureshrinker{\vspace{-0.06in}}
\label{fig:utility}
\end{figure}

\begin{figure}[tb]
\centering
\includegraphics[width=3.5in]{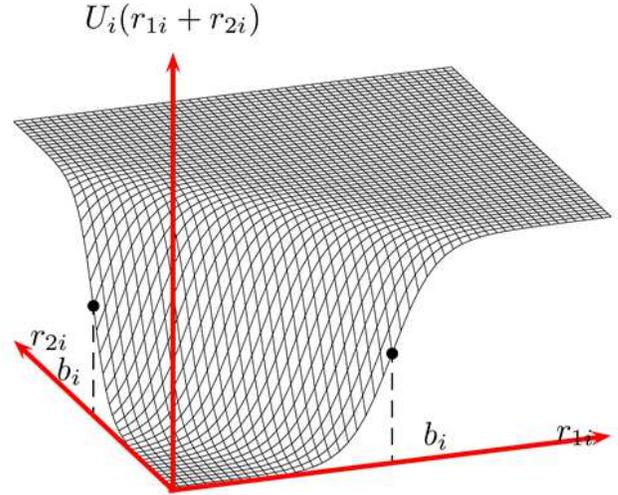}
\caption{The sigmoidal-like utility $U_i(r_{1i} + r_{2i}) = c_i(\frac{1}{1+e^{-a_i(r_{1i} + r_{2i}-b_i)}}-d_i)$ of the $i^{th}$ user, where $r_{1i}$ is the rate allocated by $1^{st}$ carrier eNodeB and $r_{2i}$ is the rate allocated by $2^{nd}$ carrier eNodeB.}
%%\myfigureshrinker{\vspace{-0.06in}}
\label{fig:Utility3D}
\end{figure}
%%%%%%%%%%%%%%%%%%%%%%%%%%
\begin {table}[]
\caption {Users and their applications utilities}
\label{table:parameters}
\begin{center}
\renewcommand{\arraystretch}{1.4} %<- modify value to suit your needs
\begin{tabular}{| l | l | l | }
%\label{table:utility}
  \hline
  \multicolumn{2}{|c|}{Applications Utilities Parameters} & \multicolumn{1}{|c|}{Users Indexes} \\  \hline
  Sig1 & Sig $a=5,\:\: b=10$  &  $i=\{1,7\}$ \\ \hline
  Sig2 & Sig $a=3,\:\: b=20$ & $i=\{2,8\}$  \\ \hline
  Sig3 & Sig $a=1,\:\: b=30$ & $i=\{3,9\}$   \\ \hline
  Log1 & Log $k=15,\:\: r_{max}=100$ & $i=\{4,10\}$   \\ \hline
  Log2 & Log $k=3,\:\: r_{max}=100$ & $i=\{5,11\}$   \\ \hline
  Log3 & Log $k=0.5,\:\: r_{max}=100$ & $i=\{6,12\}$ \\ \hline
\end{tabular}
%\caption {Should be a caption}
\end{center}
\end {table}
\subsection{Allocated Rates for $30\le R_1\le200$ and $R_2=70$}
In the following simulations, we set $\delta =10^{-3}$, the $1^{st}$ carrier eNodeB rate $R_1$ takes values between $30$ and $200$ with step of $10$, and the $2^{nd}$ carrier eNodeB rate is fixed at $R_2 = 70$. In Figure \ref{fig:ri_versus_R1}, we show the final allocated optimal rates $r_i=r_{1i}+r_{2i}$ of different users with different $1^{st}$ carrier eNodeB total rate $R_1$ and observe how the proposed rate allocation algorithm converges when the eNodeBs available resources are abundant or scarce. In Figure \ref{fig:r1i_versus_R1}, we show the rates allocated to the $1^{st}$ group of UEs by only C1 eNodeB since C2 eNodeB is not within these users range, we observe the increase in the rate allocated to these users with the increase in $R_1$. Figure \ref{fig:r1i+r2i_versus_R1} shows the final allocated rates to the $2^{nd}$ group of UEs by both C1 and C2 eNodeBs. Since these users located under the coverage area of both the macro cell and the small cell, they are allocated rates jointly using the proposed
RA with joint CA approach. Figure \ref{fig:r1i_versus_R1} and \ref{fig:r1i+r2i_versus_R1} show that by using the RA with joint CA algorithm, no user is allocated zero rate (i.e. no user is dropped). However, the majority of the eNodeBs resources are allocated to the UEs running adaptive real-time applications until they reach their inflection rates the eNodeBs then allocate more resources to the UEs with delay-tolerant applications, as real-time application users bid higher than delay-tolerant application users by using the utility proportional fairness policy.

In Figure \ref{fig:ri_versus_R1_SmallCell}, we show the rates allocated to the $2^{nd}$ group users, located under the coverage area of both the macro cell and small cell eNodeBs, by each of the two carriers' eNodeBs with the increase in the $1^{st}$ carrier eNodeB resources. In Figure \ref{fig:r1i_versus_R1_SmallCellUsers} and \ref{fig:r2i_versus_R1_SmallCellUsers}, when the resources available at C2 eNodeB (i.e. $R_2$) is more than that at C1 eNodeB, we observe that most of the $2^{nd}$ group rates are allocated by C2 eNodeB. However, the delay tolerant applications are not allocated much resources since most of $R_2$ is allocated to the real-time applications. With the increase in C1 eNodeB resources $R_1$, we observe a gradual increase in the $2^{nd}$ group rates allocated to real-time applications from C1 eNodeB and a gradual decrease from C2 eNodeB resources allocated to real-time-applications. This shift in the resource allocation increases the available resources in C2 eNodeB to be allocated to $2^{nd}$ group delay tolerant applications by C2 eNodeB.
\begin{figure}[tb]
  \centering
  \subfigure[The rates allocated $r_{1i}$ from the $1^{st}$ carrier eNodeB (i.e. the macro cell eNodeB) to users of the $1^{st}$ group (i.e. $i = 1, 2, 3, 4, 5, 6$).]{%
  \label{fig:r1i_versus_R1}
  \includegraphics[width=3.5in]{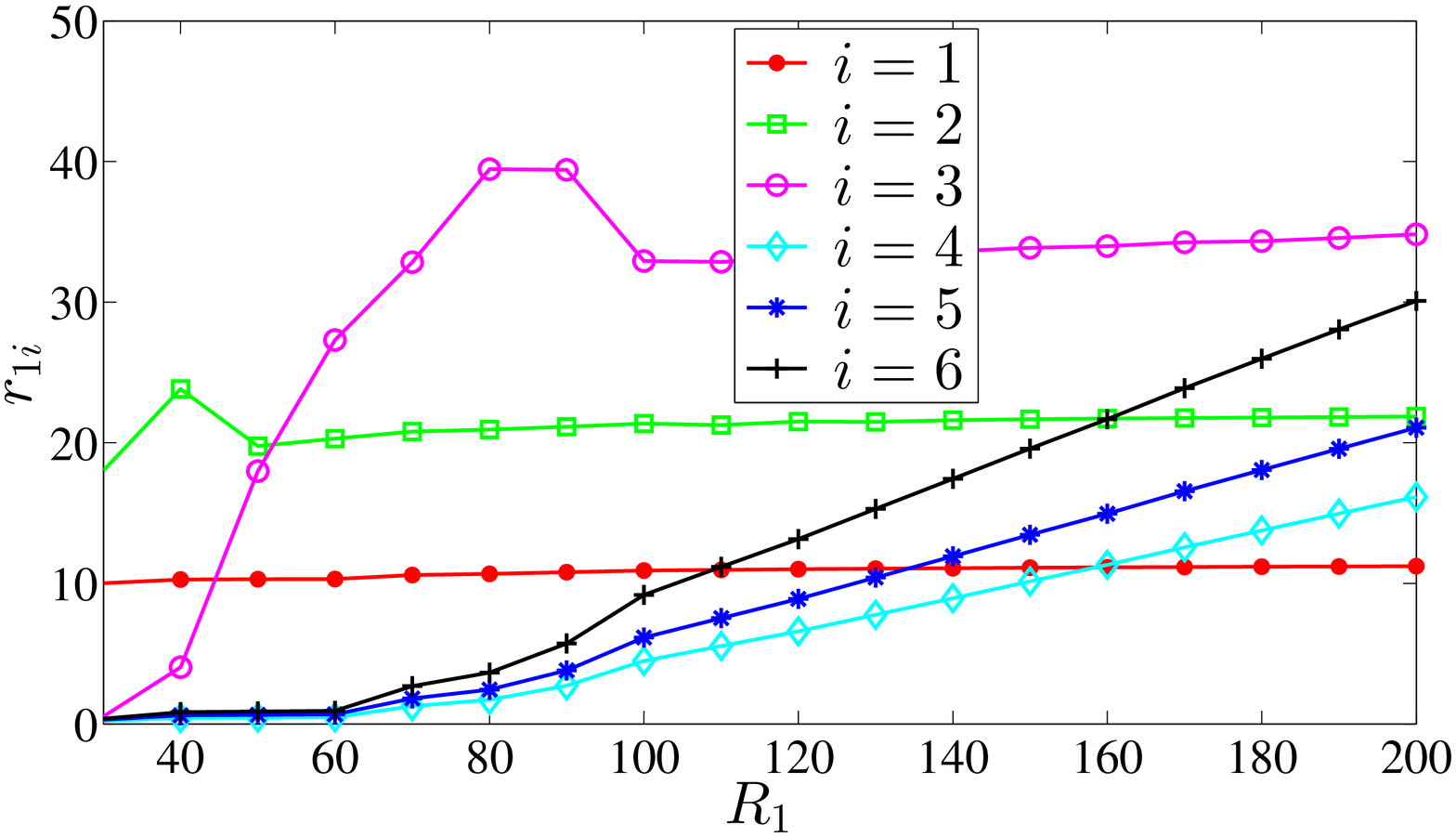}
  }\\%
\subfigure[The rates $r_{1i} + r_{2i}$ allocated from $1^{st}$ and $2^{nd}$ carriers eNodeBs (i.e. the macro cell and the small cell eNodeBs) to users of the $2^{nd}$ group (i.e. $i = 7, 8, 9, 10, 11, 12$).]{%
  \label{fig:r1i+r2i_versus_R1}
  \includegraphics[width=3.5in]{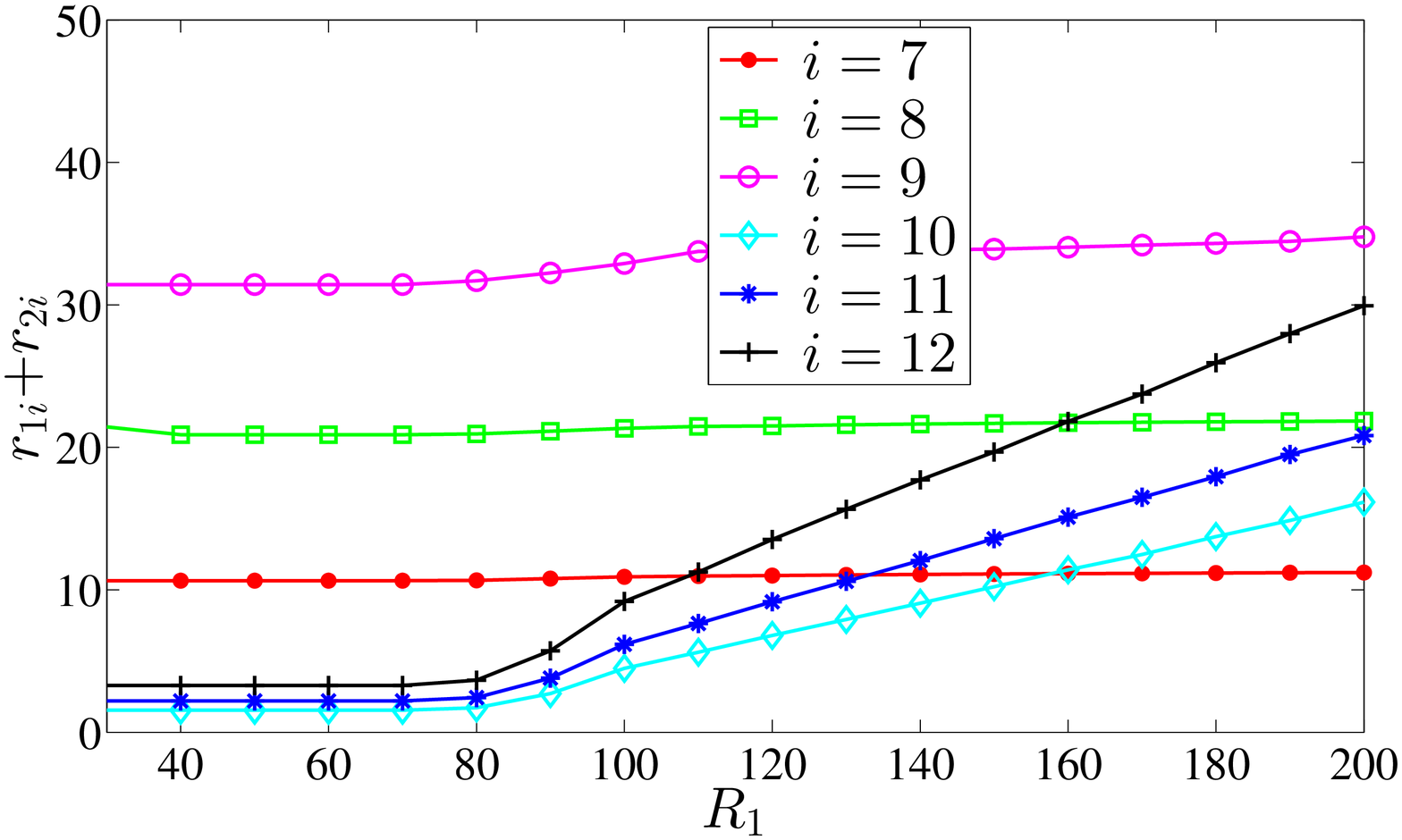}
  }%
\caption{The allocated rates $\sum_{l=1}^{K}r_{li}$ of the two groups of users verses $1^{st}$ carrier rate $30<R_1<200$ with $2^{nd}$ carrier rate fixed at $R_2=70$.}
\label{fig:ri_versus_R1}
\end{figure}

\begin{figure}[tb]
  \centering
  \subfigure[The allocated rates $r_{1i}$ from the $1^{st}$ carrier eNodeB to the $2^{nd}$ group of users.]{%
  \label{fig:r1i_versus_R1_SmallCellUsers}
  \includegraphics[width=3.5in]{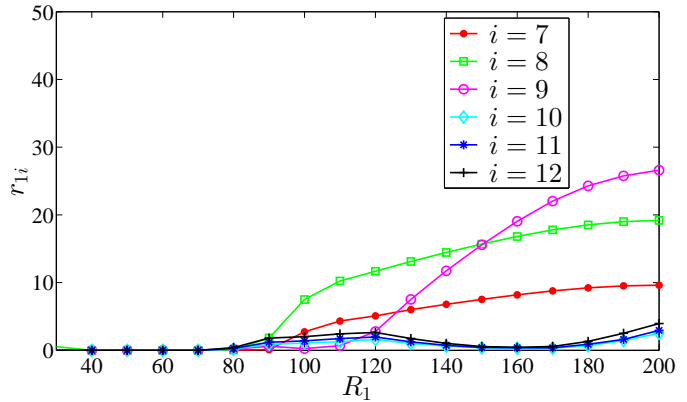}
  }\\%
\subfigure[The allocated rates $r_{2i}$ from the $2^{nd}$ carrier eNodeB to the $2^{nd}$ group of users.]{%
  \label{fig:r2i_versus_R1_SmallCellUsers}
  \includegraphics[width=3.5in]{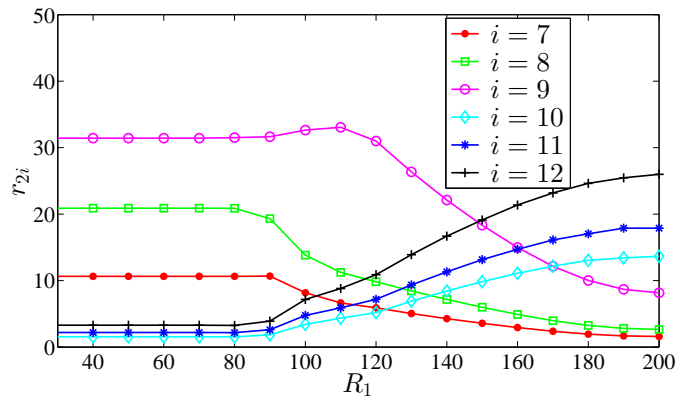}
  }%
\caption{The allocated rates from C1 and C2 eNodeBs to the $2^{nd}$ group of users with $1^{st}$ carrier eNodeB rate $30<R_1<200$ and $2^{nd}$ carrier eNodeB rate fixed at $R_2=70$.}
\label{fig:ri_versus_R1_SmallCell}
\end{figure}

\subsection{Pricing Analysis and Comparison for $30\le R_1\le200$ and $R_2=70$}
In the following simulations, we set $\delta =10^{-3}$ and the $1^{st}$ carrier eNodeB rate $R_1$ takes values between $30$ and $200$ with step of $10$, and C2 eNodeB total rate is fixed at $R_2 = 70$. As discussed before, the users' allocated rates are proportional to the users' bids. Real-time application users bid higher than delay-tolerant application users due to their applications nature and the utility proportional fairness policy. Therefore, the pricing which is proportional to the bids is traffic-dependent, i.e. when the demand by users increases, as a result the price increases and vice versa.

In Figure \ref{fig:Price_TwoMethods}, we compare between the shadow price of C1 and C2 eNodeBs when using the proposed RA with joint CA approach with their shadow prices obtained when using the multi-stage RA with CA approach in \cite{Haya_Utility1,Haya_Utility3,Haya_Utility6}. For the RA with joint CA case, we observe that the shadow price of C1 eNodeB is higher than that of C2 eNodeB for $R_1 < 80$ and approximately equal for $80\leq R_1\leq200$ which shows how it is very efficient to use the joint CA approach for the pricing of the user. We also show how the prices decrease with the increase in the eNodeBs total rate. By using this traffic-dependent pricing, the network providers can flatten the traffic specially during peak hours by setting traffic-dependent resource price, which gives an incentive for users to use the network during less traffic hours. On the other hand, for the multi-stage RA with CA approach, we show in Figure \ref{fig:Price_TwoMethods} the changes in C1 and C2 eNodeBs shadow prices
with $R_1$. When using the multi-stage RA with CA approach, all users are first allocated rates by the macro cell eNodeB, once C1 eNodeB is done allocating its resources C2 eNodeB starts allocating its resources only to the $2^{nd}$ group users as they are located within its coverage area. Since the pricing method in multi-stage RA with CA approach is not optimal, this explains why the shadow prices of C1 and C2 eNodeBs, in Figure \ref{fig:Price_TwoMethods}, when using the proposed RA with joint CA approach are less than their corresponding prices when using the multi-stage RA with CA approach. This shows how the proposed algorithm outperforms the algorithms presented in \cite{Haya_Utility1,Haya_Utility3,Haya_Utility6} as it guarantees that mobile users receive optimal price (minimum) for resources.

\begin{figure}[tb]
\centering
\includegraphics[height=2in, width=3.5in]{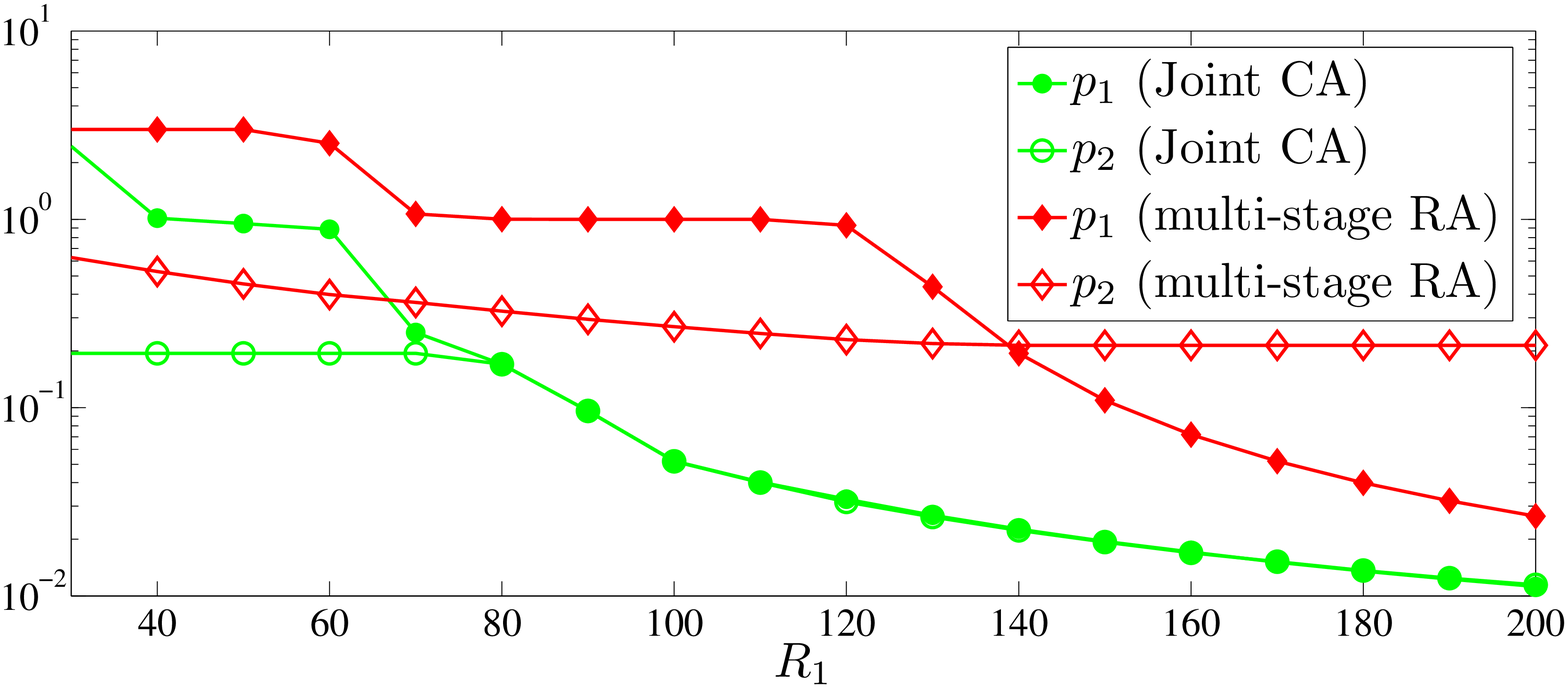}
\caption{The $1^{st}$ carrier shadow price $p_1$ and $2^{nd}$ carrier shadow price $p_2$ for both multi-stage RA with CA and joint RA methods with C1 eNodeB rate $30<R_1<200$ and C2 eNodeB rate $R_2=70$.}
%%\myfigureshrinker{\vspace{-0.06in}}
\label{fig:Price_TwoMethods}
\end{figure}

%\begin{figure}[H]
%\centering
%\includegraphics[height=2in, width=3.5in]{RatePrice_TwoMethods}
%\caption{The allocated rates from  $l^{th}$ carrier eNodeB multiplied by the $l^{th}$ carrier shadow price $p_l$ and summed up for all users $\sum_i(r_{1i}p_{1}+ r_{2i}p_{2})$ for both multiple-stage RA and joint RA methods with the $1^{st}$ carrier eNodeB rate $30<R_1<200$ and the $2^{nd}$ carrier eNodeB rate $R_2=70$.}
%%%\myfigureshrinker{\vspace{-0.06in}}
%\label{fig:RatePrice_TwoMethods}
%\end{figure}
%%%%%%%%%%%%%%%%%%%%%%%%%%%%%%%%%%%
%%%%%%%%%%%%%%%%%%%%%%%%%%%%%%%%%%%
%%%%%%%%%%%%%%%%%%%%%%%%%%%%%%%%%%%
\section{Conclusion}\label{sec:conclude}
In this paper, we introduced a novel resource allocation optimization problem with joint carrier aggregation in cellular networks. We considered mobile users running real-time and delay-tolerant applications with utility proportional fairness allocation policy. We proved that the global optimal solution exists and is tractable for mobile stations with logarithmic and sigmoidal-like utility functions. We presented a novel robust distributed algorithm for allocating resources from different carriers optimally among the mobile users. Our algorithm ensures fairness in the utility percentage achieved by the allocated resources for all users. Therefore, the algorithm gives priority to users with adaptive real-time applications while providing a minimum QoS for all users. In addition, the proposed RA with joint CA algorithm guarantees allocating resources from different carriers with the lowest resource price for the user. We analyzed the convergence of the algorithm with different network traffic densities and
presented a robust algorithm that overcomes the fluctuation in allocation during peak traffic hours. We showed through simulations that our algorithm converges to the optimal resource allocation and that the proposed algorithm outperforms the multi-stage RA with CA algorithms presented in \cite{Haya_Utility1,Haya_Utility3,Haya_Utility6} as it guarantees that mobile users receive optimal price for the allocated resources.

\bibliographystyle{ieeetr}
\bibliography{pubs}
\end{document}